%% file: main.tex
\documentclass[11pt, english]{article}
\usepackage[T1]{fontenc}
\usepackage{graphicx}
\usepackage{graphicx}
\usepackage{amsmath}
\usepackage{amsthm}
\usepackage{hyperref}
\usepackage{xspace}
\usepackage{paralist}

\usepackage{amssymb}
\usepackage{url}
\newcommand{\sortN}{\mbox{\rm Sort}\left(S,L\right)}

\newcommand{\tree}              {T}

\newcommand{\shortbound}{\frac{\volShort}{B} \log_{M/B}\frac{\volShort}{B}}
\newcommand{\longbound}{\sum_{i = 1}^{k}{\left(\frac{\volLong_{i}}{B} \log_{M/w}\frac{\volLong_{i}}{w}\right)}}

\newcommand{\bfsub}{\frac{\volLong}{B}\log_{M}\volShort + \frac{S}{B}}

\newcommand{\bfslb}{\frac{\volLong}{B}\log_{M}k}

\newcommand{\dfs}{\frac{\volLong}{w}\log_{B}k + k\log_{B}\volShort + \frac{L}{B} + \frac{S}{B}}

\newcommand{\plesl}{\mbox{\rm PLE}\left(S, L\right)}
\newcommand{\ple}{\mbox{\rm PLE}}
\newcommand{\ramsort}{\mbox{\rm 2RAMSORT}}

\newcommand{\kk}{$k$-$k$\xspace}
\newcommand{\sk}{$S$-$k$\xspace}
\newcommand{\klw}{$k$-$\tilde{k}$\xspace}

\newcommand{\lw}{\tilde{k}}  
\newcommand{\setN}              {N}

\newcommand{\intr}{I}
\newcommand{\defn}[1]{\emph{#1}}
\newcommand{\deflabel}[1]        {\label{def:#1}}

\newcommand{\thmlabel}[1]       {\label{thm:#1}}
\newcommand{\TT}{\mathcal{T}}

\newcommand{\setShort}          {\mathcal{S}}
\newcommand{\setLong}           {\mathcal{L}}
\newcommand{\volShort}          {S}
\newcommand{\volLong}           {L}
\newcommand{\numShort}          {S}

\newcommand{\lenLong}           {w}
\newcommand{\numStripe}         {k}
\newcommand{\eqlabel}[1]        {\label{eq:#1}}

\usepackage{bbm}
\usepackage{multicol}
\usepackage{multirow}
\usepackage{babel}
\usepackage{algorithm}
\usepackage{algorithmic}
\floatname{algorithm}{Procedure}

\newcommand{\thmref}[1]{Theorem~\ref{thm:#1}}
\newcommand{\lemref}[1]{Lemma~\ref{lem:#1}}
\newcommand{\lemlabel}[1]       {\label{lem:#1}}

\newtheorem{theorem}{Theorem}
\newtheorem{observation}[theorem] {Observation}
\newtheorem{lemma}[theorem] {Lemma}
\newtheorem{definition}[theorem] {Definition}

\usepackage{wrapfig}

\usepackage{enumitem}
\setitemize{noitemsep,topsep=0pt,parsep=0pt,partopsep=0pt}
\author{
  Michael A. Bender\thanks{
   Stony Brook University,
   \texttt{bender@cs.stonybrook.edu}. This work was supported in part by NSF grants 
CCF-1725543, 
CSR-1763680,  
CCF-1716252, 
CCF-1617618, 
CNS-1938709,   
and by Sandia National Laboratories.}
  \and
   Mayank Goswami\thanks{
   Queens College, CUNY,
   \texttt{mayank.goswami@qc.cuny.edu}. Supported by NSF grants CRII-1755791 and CCF-1910873.}
  \and
  Dzejla Medjedovic\thanks{
   International University of Sarajevo,
   \texttt{dzmedjedovic@ius.edu.ba}.}
  \and
 Pablo Montes\thanks{
  Google Inc.,
  \texttt{pabmont@gmail.com}.}	
  \and
  Kostas Tsichlas\thanks{
   Aristotle University of Thessaloniki,
   \texttt{tsichlas@csd.auth.gr}.}
}

\date{}
\title{Batched Predecessor and Sorting with Size-Priced Information in External Memory}

\begin{document}
\maketitle              
\begin{abstract}
In the unit-cost comparison model, a black box takes an input two items and outputs the result of the comparison. Problems like sorting and searching have been studied in this model, and it has been generalized to include the concept of priced information, where different pairs of items (say database records) have different comparison costs. These comparison costs can be arbitrary (in which case no algorithm can be close to optimal (Charikar et al. STOC 2000)), structured (for example, the comparison cost may depend on the length of the databases (Gupta et al. FOCS 2001)), or stochastic (Angelov et al. LATIN 2008). Motivated by the database setting where the cost depends on the sizes of the items, we consider the problems of sorting and batched predecessor where two non-uniform sets of items $A$ and $B$ are given as input.

(1) In the RAM setting, we consider the scenario where both sets have $n$ keys each. The cost to compare two items in $A$ is $a$, to compare an item of $A$ to an item of $B$ is $b$, and to compare two items in $B$ is $c$. We give upper and lower bounds for the case $a \le b \le c$, the case that serves as a warmup for the generalization to the external-memory model. Notice that the case $b=1, a=c=\infty$ is the famous ``nuts and bolts'' problem. 

(2) In the Disk-Access Model (DAM), where 
transferring elements  between disk and internal memory is the main bottleneck,
we consider the scenario where elements in $B$ are larger than elements in $A$. The larger items take more I/Os to be brought into memory, consume more space in internal memory, and are required in their entirety for comparisons.

A key observation is that the complexity of sorting depends heavily on the interleaving of the small and large items in the final sorted order. If all large elements come after all small elements in the final sorted order, sorting each type separately and concatenating is optimal. However, if the set of predecessors of $B$ in $A$ has size $k\ll n$, one must solve an associated batched predecessor problem in order to achieve optimality.

We first give output-sensitive lower and upper bounds on the batched predecessor problem, and use these to derive bounds on the complexity of sorting in the two models. Our bounds are tight in most cases, and require novel generalizations of the classical lower bound techniques in external memory to accommodate the non-uniformity of keys.

\end{abstract}

\input{intro.tex}
\input{ramversion.tex}
\input{damversion1.tex}

\input{damversion2.tex}

\input{conclusion.tex}

\bibliographystyle{splncs04}
\bibliography{latin-sorting.bib}
\end{document}

%% file: intro.tex

\section{Introduction}

In most published literature on sorting and other comparison-based
problems (e.g., searching and selection), the traditional assumption is
that a comparison between any two elements costs one unit, and the
efficiency of an algorithm depends on the total number of comparisons
taken to solve the problem. In this paper, we study a natural extension
to sorting, where the cost of a comparison between a pair of elements can
vary, and the comparison cost is the function of the elements being
compared. 

We work in both the  the random-access-machine (RAM)  and the
disk-access-machine (DAM)~\cite{AggarwalVi88} models (described below).
We derive worst-case upper and lower bounds for comparison-based sorting and batched predecessor.

In the RAM model, we assume that comparisons between a pair of keys have
an associated cost that depends on the ``type'' of keys involved. As a
toy problem, consider the case when we have keys of two types---$n$ red
keys and $n$ blue keys. A comparison between a pair of red keys costs
$a$, between a red key and a blue key costs $b$, and between a pair of
blue keys costs $c$. Without loss of generality we can assume $a<c$, which gives rise to three cases to be considered, $a<b<c$, $a<c<b$, and $b<a<c$ (when $b=1$ but
$a=c=\infty$ corresponds to the well-known nuts and bolts problem
\cite{AlonBlFi94}.) Traditionally such problems have been studied in the
context of priced information \cite{charikar2000query,cicalese2005new},
where the cost of an algorithm is studied in the competitive analysis setting. These comparison costs can be arbitrary (in which case no algorithm can be close to optimal \cite{charikar2000query}), structured (for example, the comparison cost may depend on the length of the databases \cite{gupta2001sorting}), or stochastic \cite{angelov2008sorting}. 

In this paper we consider the setting of \cite{gupta2001sorting}, where the price of the information depends on the length of the keys being compared. However, we depart from the competitive analysis model by considering the worst-case cost,
but parameterized by the specific distribution (or the ``interleaving'')
of the elements in the final sorted order.

Then we turn to the \emph{disk access machine} (\emph{DAM}) model (also
called the \emph{external-memory model} or the \emph{I/O
model})~\cite{AggarwalVi88}. This model captures an essential aspect of
modern computers---that computation is fast but transferring data between
levels of a memory hierarchy is slow.  Data is transferred from an
infinite external disk to a RAM of size $M$ in blocks of size $B$;
the cost of the algorithm is measured by the number of block transfers
(I/Os) that it uses.

In the DAM model the notion of comparison cost naturally comes into play
when elements have different \emph{sizes} (or \emph{lengths}). In this
model, comparisons come for free once the elements are in RAM.
However, it is cheaper to transfer
short elements into RAM than long elements.  For example,  if a key has
length $w$, where $w\leq B$, then up to $B/w$ keys can be fetched with
one I/O; similarly, if $w\geq B$, then it takes $w/B$ I/Os to bring that
key into memory.
Moreover, a long element, when brought into RAM, will displace
a larger volume of keys than a short element.\footnote{Note that the 
DAM model actually models the memory transfers between any two levels 
of the memory hierarchy. Although this paper adapts the terminology 
of I/Os between RAM and disk, the model also 
applies to cache misses between cache and RAM. 
In the former case, elements could be larger than $B$ but are (essentially always) 
much smaller than $M$. In the latter case, elements 
could have a length that is a nontrivial fraction of $M$.}

Consider the following generalization of the RAM problem to DAM: we are
given $S$ keys of unit size (short keys) and $L/w$ (long) keys of size
$w$ each (total volume $L$). 
%
%
What is the optimal sorting algorithm for when there are two key sizes?  
We want to express our results parameterized by the interleaving
of the elements in their final sorted order.  Let interleaving parameter
$k$ denote the number of consecutive runs of large keys (\emph{stripes})
in the final sorted order. In other words, the set of predecessors of $L$ in $S$ has size $k$.  We want to express the performance of
the sorting algorithm, as a function of $S$, $L$, $w$, and $k$.   

Sorting with two key lengths helps illustrate a special connection
between sorting and batched searching. Consider the following batched searching
problem, which we call the \emph{PLE} (placement of large elements)
problem.  We have $S$ keys of unit size (short keys), which are given in
sorted order.  We have $L/w$ (long) keys of size $w$ each, and the
objective is to find which short key is the immediate predecessor of each
long key.  The PLE problem is a lower bound on the sorting problem
because it starts off with more information than the original sorting and
asks to do less.  

Often, the complexity of PLE dominates the complexity of sorting.
However, obtaining lower bounds on the PLE presents several challenges.
First, for many reasons (we want bounds in terms of $k$, keys have
different sizes, and there are many searches happening in batch)
standard information-theoretic lower bounds bounding how much information
is learned do not immediately apply: different I/Os can ``learn'' very
different amount of information.  Second, because of the nonuniformity of the keys, the complexity of PLE turns out
to be a minimum of two terms, each optimal for a certain range of values
of $w$ (the length of large elements).  Each case requires different
techniques.  Third, we have to take preprocessing into account. PLE is a
batched searching problem with a nontrivial preprocessing-query
tradeoff \cite{BenderFaGo14}. However, in our context it is a subproblem of sorting, and
other parts of sorting dictate how much preprocessing is allowed.

\paragraph{Related Work.} In RAM, algorithms that work with inputs with priced
information have been studied before\cite{charikar2000query,cicalese2005new,gupta2001sorting,angelov2008sorting}.In this setting the results are presented using a competitive analysis
framework.  Another example of varying comparison costs is the well-known
nuts-and-bolts problem~\cite{AlonBlFi94}. Interleaving-sensitive lower bounds
and batched searching in RAM are related to lower bounds for  sorting
multisets~\cite{MunroSp76} and distribution-sensitive
set-partitioning~\cite{Elmasry05}.

Aggarwal and Vitter~\cite{AggarwalVi88} introduced the external-memory (DAM)
model.  The lower bounds they establish for fundamental comparison-based
problems were generalized by Arge, Knudsen, and Larsen~\cite{ArgeKnLa93}, and by
Erickson~\cite{Erickson05} to the external algebraic decision tree model.
Prominent examples studying lower bounds on batched and predecessor searching
are found in \cite{Arge03,ArgePrRa98,BenderFaGo14}.

Most previous work that considers variable-length keys does so in the context
of B-trees~\cite{McCreight77,DiehrFa84,LarmoreHi85,PinchukSh92,BenderHuKu10}. 

\paragraph{Relation to string sorting.} Arge et al.~\cite{ArgeFeGr97} study the I/O complexity
of sorting strings in external memory. The authors consider different models of key divisibility and derive upper and lower bounds for each model.
The problem is different from ours because strings are not atomic: they can
be broken into their individual characters which can reduce the I/O complexity
of sorting. 

Many systems are designed to be consistent with the notion of indivisible keys.
For example, sorting and searching libraries such as GNU Sort~\cite{gnu-sort}
or Oracle Berkeley DB~\cite{berkeley-db}  allow (or require) one to pass in a
comparison function as a parameter. Note that the algorithms and lower bounds
for indivisible keys in this paper also hold as the worst-case lower bounds for
string sorting --- when the entire string is necessary in memory in order to
break the tie.

\paragraph{Organization.} In Section $2$ we present the  RAM version of
our problem. We present the sorting problem in external-memory in
Section $3$ and relate it to the batched predecessor problem. We discuss the challenges in extending the RAM solution to
this case in section $4$. 

We then derive lower and upper bounds on the batched predecessor problem in sections $5$ and $6$, respectively. We end with open problems in Section $7$. Due to space constraints, complete proofs are relegated to the appendix.

%% file: ramversion.tex
\vspace{-1mm}
\section{Warmup: the RAM version}
\vspace{-1mm}
\noindent\textbf{Two types, RAM version (\boldmath $\ramsort$).} The
input is $n$ red and $n$ blue keys, and the output is the sorted sequence
of all keys. A comparison between a pair of red keys costs $a$, between a
red key and a blue key costs $b$, and between a pair of blue keys costs
$c$. Without loss of generality we can assume that $a<c$. 

\noindent\textbf{Interleaving-sensitive analysis.} The optimal sorting
cost in RAM depends on the final interleaving of the elements in the
final sorted order. If in the final sorted order all red keys come before
all blue keys, then $\Theta(an \log n + c n \log n)$ is the optimal total
comparison cost, because the algorithm that separately sorts and
concatenates is optimal.  However, if the red and blue keys alternate in
the final sorted order, then no blue-blue comparisons are ever required
to sort.

\noindent\textbf{Stripes and the interleaving parameter \boldmath $k$.} A
consecutive run of red or blue keys in the final sorted order is called a
\textit{stripe}. Define $k$ to be the number of blue stripes, and let $\ell_{i}$ (respectively $s_{i}$)
be the number of blue (respectively red) keys in stripe $i$. The notation $\ell$ and $s$ are chosen to correspond with the later sections when red elements will be small and blue elements will be large.

\begin{theorem}\label{RAM}
$\ramsort$ has the following comparison cost complexity for the case $a \le b \le c$:

$$\Theta ( a n \log n + b(k \log n + n \log k) + c \sum_{i=1}^{k} \ell_{i} \log \ell_{i})$$
\end{theorem}

\begin{proof}
We are interested in the version that is most relevant to us, the case when $a<b<c$. This is because red elements can be considered small, blue elements can be considered large, and under the natural setting where comparisons involving red elements cost less than those involving blue elements, we get that $a<b<c$.

\noindent\textbf{Lower bounds.} First, the number of permutations any algorithm for $\ramsort$ must achieve is at least $n!$. Any comparison reduces these by a factor of at most $2$, and the cheapest comparison costs $a$, thus giving a $a \log n! = \Omega(a n\log n)$ lower bound.

Second, consider the instance where the red elements are already sorted for free, the stripes of blue elements are already provided for free, and one is required to finish sorting. In this case, no comparison involving a red element is useful, the only comparisons available cost $c$, and these must be used to sort the contents of each stripe separately. This gives us a lower bound of $\Omega(c \sum_{i=1}^{k} \ell_{i})$.

Proving the second term as a lower bound involves the batched predecessor problem. Consider the instance where the red elements are sorted for free, and the algorithm is just required to discover the content of the $k$ blue stripes. This is identical to the batched predecessor problem on the blue elements, and any sorting algorithm also solves this subproblem. 

There are at least ${n \choose k} S(n,k)$ permutations to consider, where $S(n,k)$ is the \textit{Stirling number of the second kind}\footnote{this is the number of ways to partition a set of size $n$ into $k$ non-empty, disjoint subsets.}. This is because there are ${n \choose k}$ positions to place the blue stripes among the red elements in, and $S(n,k)$ distinct possibilities for the contents of the blue stripes.

An easy lower bound on $S(n,k)$ is $k^{n-k}$, which is achieved by fixing an ordering of the $n$ blue elements, sending the first $k$ to a distinct subset, and now for the remaining $n-k$ blue elements one has $k$ possibilities for each one.

Thus we get that the total number of permutations is at least ${n \choose k}S(n,k)$, which is at least $(n/k)^{k} k^{n-k}$. Since red elements are already sorted, comparisons of cost $a$ are useless, and the cheapest available comparisons are those costing $b$. Thus we get a lower bound of $\Omega(b (k \log (n/k) + (n-k) \log k))$, which equals $\Omega(b (k \log n + n \log k - 2k \log k))$ .

If $k \leq n/3$, $n \log k - 2k \log k \geq \frac{n}{3} \log k$, and we get the claimed lower bound. If $k> n/3$, we claim that $b n \log n$ is a lower bound, which matches our lower bound for this case. Consider the instance where red elements are sorted, $k$ blue representatives of each stripe are given (unsorted) and one is told that the remaining $n-k$ blue elements are all larger than the largest red elements. To finish the batched predecessor problem, there are at least ${n \choose k} k!$ permutations to check (find where to place the representatives, find the order in which representatives should go). No comparison between red elements are useful, the cheapest available comparison costs $b$, so we get a  lower bound of $b \log ({n \choose k} k!)$, which since $k>n/3$, is $\Omega(b n \log n)$.

Since we derived the lower bound on a constant number of instances of $\ramsort$, we can claim a lower bound of the maximum complexity of these instances, which in turn is the same as the sum of the complexities in $\Omega(.)$ notation, since we only have constantly many such instances.

\noindent\textbf{Upper bounds for the case $a<b<c$.} The algorithm proceeds in a similar
fashion:
\begin{enumerate}
\item Sort the $n$ red elements using cost $a$ comparisons.

\item  Build a binary tree $\mathcal{T}_{r}$ on the sorted red elements. The algorithm also maintains
a binary tree $\mathcal{T}_{b}$ (initially empty) on the set of discovered “border” of
red elements, i.e., the red elements which immediately precede and succeed
a discovered stripe.
A blue element is first sent down $\mathcal{T}_{b}$ to find whether it belongs to an already
discovered stripe. If it does not belong to an already discovered stripe, it is
bound to discover a new stripe, and it is sent down $\mathcal{T}_{r}$. The new bordering
red elements are then inserted into $\mathcal{T}_{b}$. It is clear that only $k$ blue elements
go down $\mathcal{T}_{r}$, whereas every blue element may go down $\mathcal{T}_{b}$, which has at
most $k$ leaves. All the comparisons in this step are cost $b$.

\item Sort the stripes of blue elements using type c comparisons.
\end{enumerate}
One can easily verify that the running time of this algorithm matches the one
in the theorem statement.
\end{proof}

%% file: damversion1.tex
\vspace{-2mm}\section{Sorting and Batched Predecessor in External Memory with Size-Priced Information}

The input to the two-sized sorting and batched predecessor problems are $S = \{s_{*}\}$ (the small records) and $L=\{\ell_{*}\}$ (the large records, each of size $1 < w \le M/2$). A set of large elements forms a \defn{stripe} if for each pair of large elements $\ell_{i}$ and $\ell_{j}$ in the stripe, there does not exist a small element between $\ell_{i}$ and $\ell_{j}$ in the final sorted order. Let $k$ be the number of large-element stripes, and let the large-element stripes be $\setLong_1,\setLong_2,\ldots, \setLong_k$, as they are encountered in the ascending sorted order. The parameters in the complexity analysis of sorting and batched predecessor are thus $S$, $L$, $w$, $k$, and $\{L_{i}\}_{i=1}^{k}$.

\begin{definition}[\defn{Two-Sized Sorting $\sortN$}] 
The input is an (unsorted) set of elements $N = S \cup L$.  
Set $S$ consists of $S$ unit-size elements, and  $L$ consists of $L/w$ elements, each of size $\lenLong$, where\footnote{We overload notation  for 
convenience of presentation. We assume $w \ge B$ also  for the convenience of presentation. Our bounds hold for any $1 < w \le M/2$; we extend our results to this entire range in Appendix~\ref{allcases}.} $B \le \lenLong \le M/2$. 
The output comprises the elements in $\setN$, sorted and stored contiguously in external memory.
\deflabel{sort}
\end{definition}

\begin{definition}[\defn{$\ple$-Placement of Large Elements:}] The input is the \emph{sorted} set of small elements $\setShort = \{s_{1}, s_{2}, \ldots, s_{\numShort}\}$, and
the \emph{unsorted} set of large elements $\setLong = \{\ell_{1}, \ell_{2}, \ldots, \ell_{L/w}\}$. In the output, elements in $\setShort$ are sorted, and elements in $L$ are sorted according to which stripe they belong to, but arbitrarily ordered within their stripe. 
\end{definition}

The following theorem relates the complexities of the sorting and the batched predecessor problem.

\begin{theorem}[\textbf{Sorting complexity}]
\thmlabel{sortingcomplexity}
Denote by $\plesl$ the complexity of the $\ple$ problem. Then the I/O complexity of Two-Sized Sorting $\sortN$ is
 
$$\Theta\left(\shortbound + \plesl+  \left(\longbound + \frac{L}{B} \right)\right).$$
\end{theorem}

We will use this section to prove the first and the third term of the \thmref{sortingcomplexity}. It is easy to see that $\plesl$ is an instance of $\sortN$ that starts with more information (sorted small elements), and requires less (just the contents of the stripes, unsorted).

\begin{lemma}
\lemlabel{sortshort}
$\sortN = \Omega \left( \frac{\volShort}{B} \log_{M/B} \frac{\volShort}{B}\right)$. 
\end{lemma}

\begin{proof}
We denote by $p$ the total number of permutations that any algorithm for Two-Sized Atomic-Key Sorting must distinguish between in order to  sort. We bound $p$ in terms of the number of large-element stripes $k$:
$$p \ge \frac{\volShort !}{(B!)^{\volShort/B}} {\volShort - 1\choose k} \left(\frac{L}{w}!\right) {L/w-1 \choose k-1}.$$ 

The four factors that comprise the right side include (1) sorting $S$ (after
sorting within the small blocks, the total number of permutations goes down by a
factor of $(B!)^{S/B}$), (2) choosing the $k$ locations for stripes within
$S$, (3) sorting $L$ and (4) forming $k$ large-element stripes
by choosing $k-1$ delimiters in the sorted $L$.  


We can assume that the elements in memory are sorted at all times, because
maintaining this order requires no additional I/Os. 
The number of remaining permutations goes down by at most ${M \choose B}$ after one memory transfer, thus a lower bound on the number of I/Os to sort is
$\mbox{\sc sort-lower}({N}) = \Omega\left(\log p/{\log {M\choose B}}\right)$. 
Using that $p \ge \frac{\volShort !}{(B!)^{\volShort/B}}$, and that 
$\log {M \choose B}=\Theta(B \log (M/B))$, 
we get
$$\mbox{\sc sort-lower}({N}) = \Omega \left( \frac{S \log S - S \log B}{B \log (M/B)} \right) = \Omega \left( \frac{\volShort}{B} \log_{M/B} \frac{\volShort}{B}\right),$$ 
concluding the proof.

\qed
\end{proof}

The third term concerns sorting the stripes $L_{i}$ of large elements. Although this is same-size sorting, we need a slight generalization of the classical Aggarwal and Vitter \cite{AggarwalVi88} result for records of size $w>1$. 

\begin{lemma}[Aggarwal and Vitter]
\lemlabel{genlemma}
Consider an external-memory algorithm $A$ that sorts the total volume $V$ of $V/w$ elements, each of size $w$.
\begin{enumerate}
\item If $1 \leq w < B$, ${A}$ requires $\Omega \left( \frac{V}{B} \log_{M/B} \frac{V}{B} \right)$ block transfers.
\item If $w \geq B$, ${A}$ requires $\Omega \left( \frac{V}{B} \log_{M/w} \frac{V}{w}\right)$ block transfers.
\end{enumerate}
\end{lemma}

\noindent\textbf{Remark:} The third term in the bound in \thmref{sortingcomplexity} is derived by substituting $V=L_{i}$ in the above lemma, and adding all the lower bounds (since sorting stripe $i$ is independent of sorting stripe $j$).

\begin{proof}
In both cases, we count the total number of possible output permutations and the maximum permutations achievable during a single I/O, or during the input of one element ($w/B$ I/Os), whichever is larger.
\begin{enumerate}
\item $1 \leq w < B$: Assume that $w$ divides $B$. In a linear scan, we can internally sort every block, which restricts the possible output permutations to 
\[  \frac{(V/w)!} { ((B/w)!)^{V/B}  }. \]

When a block is input, there are at most $(M-B)/w$ sorted elements in memory. The incoming block contains $B/w$ sorted elements, so the number of remaining output permutations reduces by at most a factor of  
\[{ (M-B)/w + B/w \choose B/w} = { M/w \choose B/w}.\]

Thus we get that the algorithm requires at least 

\[ \Omega \left(\frac{\log \left( (V/w)! / ((B/w)!)^{V/B} \right)}  {\log { M/w \choose B/w}} \right) \] block transfers. Using the same bounds for ${n \choose k}$ as in \lemref{sortshort}, we get the desired bound. 

\item  $w \geq B$: Assume for simplicity that $w$ is an integer multiple of $B$. In this case, there are $(V/w)!$ possible output permutations. One can scan every chunk of size $M$, but this does not change the bound we present asymptotically ( $\log_{M/w} V/w$ changes to $\log_{M/w} V/M$).

When an element is input, there are at most $(M-w)/w = (M/w)-1$ (sorted) elements in memory. The input of an element costs $w/B$ I/Os, and this element can go into any one of $M/w$ positions between the elements in memory. Hence the maximum branching factor for one element input is $M/w$. 

This implies that the number of element inputs is 
$$\Omega\left(\frac{(V/w)\log (V/w)}{\log (M/w)}\right) = \Omega\left(\frac{V}{w} \log_{M/w} \frac{V}{w}\right),$$
and multiplying by the cost of every large element input ($w/B$) gives us the claimed bound. 
\end{enumerate}
\qed
\end{proof}

As in the RAM setting, since we have three subproblems, their maximum complexity, and hence the complexity of their sum, is a lower bound on $\sortN$. We have thus reduced the sorting problem to the batched predecessor problem, which will occupy the rest of this article.

\vspace{-2mm}\section{Main Challenges in the Batched Predecessor Problem}

For $\plesl$, one can see that there is not much point
comparing large records to each other; one would rather compare a large record
to more small records than one large record. We need the notion of a
\textbf{fan-out}, which measures the efficiency of an I/O. In $\plesl$, large elements are the ones trying to find their locations amongst
the small elements. A large element is called active during an I/O if it is
either in memory or in the block transferred during this I/O. Before an I/O,
any active large element has a set of locations where it might lie, and this
set gets reduced by a certain factor (possibly $1$) after this I/O. The fan-out
of an I/O is defined to be the product of all such factors for all large
elements active during this I/O.  

We now describe the three main challenges in
extending the RAM solution to external memory.

\noindent\textbf{1.Non-uniformity:} In the unit-sized setting, the transfer of a block to main memory can decrease the number of permutations to be checked by a factor of at most\footnote{The proof of the lower bound for sorting $N$ unit-sized keys in \cite{Erickson05} proceeds in the following fashion: assuming that all blocks are sorted (using a linear scan costing $N/B$ I/Os), there are $N!/(B!)^{(N/B)}$ permutations required to achieve, and the transfer of a block of $B$ sorted elements into the main memory containing $M-B$ sorted elements can at most an ${M \choose B}$ fraction of these permutations (the ``fan-out,'' since this is the degree of the node in the decision tree). Standard algebra gives a lower bound of $\frac{N}{B} \log_{M/B} \frac{N}{B}$.} $B! {M \choose B}$. In our setting, the number of comparisons performed by an I/O varies depending on whether the block transfer carries large records or small records, and what the contents of RAM are at the time of the I/O. 

\begin{itemize}
    \item The transfer of a large element into main memory full of large elements gives only ${M/w \choose 1}$ per $w/B$ I/OS as a large-element transfer costs $w/B$.
    \item The transfer of $B$ small records into main memory filled with small records gives ${M \choose B}$. 
    \item The transfer of a large element into a memory full of $M$ small elements gives a fan-out of $M+1$. 
    \item While the above three cases are tight, the main issue is in getting \textit{an upper bound on how much a small block I/O can achieve}. The main memory can hold $p=(M-B)/w$ large elements, and an incoming small
block has $B$ small elements. Thus naively the maximum fan-out can be upper
bounded by $B^p$, which is not tight. Our main aim is to get a better upper
bound on this fan-out.
\end{itemize}
Both our upper and
lower bounds are a minimum of two terms, where one
dominates the other depending on how large the large elements are (whether they
can be brought into memory multiple times or just once).

\noindent\textbf{2. Requiring output-sensitive lower bound limits adversarial arguments:} Lower bounds on the
unit-sized batched predecessor problem
in external memory were recently obtained in \cite{BenderFaGo14}. The adversary
strategy in the comparison model was quite simple since the adversary had the
freedom to place the elements being searched for at any place in the sorted
set. By maintaining the invariant that all the elements being searched for
currently in main memory must have disjoint search spaces, it was able to
guarantee a fan-out of at most $2^B$, and in some cases, a fan-out of at most
$B$. 

In our setting, a more complicated adversarial analysis is required that forms exactly $k$ stripes at the end. Using this, we can argue a fan-out of at most $2^B$ on \textit{most} small block I/Os.

\noindent\textbf{3. Have to take preprocessing into account:} $\ple$ is not a
traditional searching problem, but a subproblem of sorting. We cannot be
concerned only with the query time, but also the preprocessing time: what is
the minimum amount of preprocessing needed to achieve a given query time? Even
for the classical single-element-search (for which the well known $B$-tree
provides optimal query time), this is, to the best of our knowledge, not
known. We have the following observation  which might be surprising at a first glance.

\begin{observation}\label{pablo}
Given a sorted array of $N$ keys on disk, there exists an algorithm that uses extra space of $O(N^{(1-1/\log B)}/B) (= o(N/B))$ blocks, and answers single-element search query in (optimal) $O(\log_{B} N)$ I/Os. If the query time is required to be at most $c \log_{B} N$, then any algorithm needs to preprocess $\Omega(N^{(1-(2c/\log B))}/B)$ extra blocks.
\end{observation}

\begin{proof}
We first prove the lower bound. We can assume that the element $x$ being searched for is always inside the memory at all times. The $N$ keys on disk are stored in $N/B$ blocks; call the set of these blocks $\mathcal{B}$. Let $S$ denote the current \textit{search space} of $x$: this is the set of locations in $N$ that $x$ can lie in, given all the information achieved by the algorithm until now. Any algorithm that solves this problem is described by a decision tree which
has nodes corresponding to I/Os of either a preprocessed block or a block from $\mathcal{B}$. The decision tree has at least $N/B$ leaves.

We first relax that the algorithm only locate $x$ to within a space of $2B$, i.e., once $ \vert S \vert \leq 2B$, we will give the algorithm the exact position for free. With this, the input of a block from $\mathcal{B}$ can reduce $S$ by a factor of at most $4$ (actually, this factor is $2S/(S-B)$, which is very close to $2$ when $S$ is large, and equals $4$ when $S=2B$). Since the query time cannot exceed $c\log_{B} N$, $S$ can reduce by a factor of at most $4^{c \log_{B}N} = N^{1-(2c/\log B)}$. This still leaves a factor of $N^{(1-(2c/\log B))}/B$ to account for.

Call the set of all extra blocks preprocessed by the algorithm $\mathcal{B}^{'}$. Let $K$ be the set of all elements in blocks in $\mathcal{B}^{'}$ such that no two elements in $K$ belong to the same block in $\mathcal{B}^{'}$. The following holds:

\begin{enumerate}
\item To preprocess $\mathcal{B}^{'}$, the algorithm required $\Omega(K)$ I/Os.

\item The maximum fan-out that using preprocessed blocks can achieve is $O(K)$.
\end{enumerate}

Proof of 1: Since there are $K$ elements from different blocks, each of these
blocks needed to be inputted at least once at some point of preprocessing, hence requiring at least $K$ memory transfers.

Proof of 2: Regardless of
the choice of K elements within $S$, there will be at least one gap
that is of size at least $S/K$. Hence, the maximum
factor by which the search space can be decreased is $K$.

Since the decision tree still must have enough nodes to guarantee the remaining fan-out of $N^{1-(2c/\log B)}/B$, the above two observations imply that $K = \Omega(N^{1-(2c/\log B)}/B)$, thus finishing the proof.

\noindent\textbf{Upper Bound}: Let $\beta = N^{1/\log B}$. In a linear scan, the algorithm can write out every $\beta$th element from $N$, and store them in contiguous blocks. There are $N^{1-1/\log B}$ such elements, and this requires and extra space of $N^{(1-1/\log B)}/B$ blocks. The algorithm then builds a $B$-tree on this set of elements. 

The search proceeds by first going through the $B$-tree, until the search space of $x$ is reduced to a set of size $\beta$. On this set, the algorithm performs a simple binary search. The total runtime is bounded by $\log_{B} N^{(1-1/log B)} + \log_{2} \beta \leq 3 \log_{B} N$ I/Os.
\qed
\end{proof}

%% file: damversion2.tex
\vspace{-3mm}\section{Complexity of the Batched Predecessor problem: Lower Bounds}
In this section, we prove the lower bounds for the $\plesl$ problem in the following theorem:

\begin{theorem}[\textbf{PLE Lower Bound}]
\thmlabel{plelower}
 $$\plesl = \Omega\left( \min\left\{ \frac{kw}{B} \log_{M} S+ \bfslb, \frac{k}{B}\log S +\frac{\volLong}{wB}\log k  + \frac{L}{B}\right\}\right).$$
\end{theorem}

In order to prove \thmref{plelower}, first divide the $\plesl$ problem further into three subproblems. Doing this helps us develop a more intricate adversarial analysis that gives us tight lower bounds. We then develop matching upper bounds on the $\plesl$ problem. We consider the following three subproblems of $\ple$, whose complexities lower bound the complexity of $\ple$, and hence $\sortN$:

\begin{compactenum}
\item \sk: An instance with only \emph{one large element} in each large-element stripe.
\begin{compactitem}
\item Input: Set $S$ of unit-sized elements $s_{1},\ldots,s_{\volShort}$ \defn{(sorted)}, where $s_1=-\infty$ and $s_{\volShort}=\infty$, and large elements $\ell_{1},\ldots,\ell_{k}$ (volume $kw$) \defn{unsorted}. 
\item Output: For each $\ell_{i}$ output $s_{j}$ such that $s_{j} \leq \ell_{i} \leq s_{j+1}$. It is guaranteed that no other $\ell_{k}$ satisfies $s_{j} \leq \ell_{k} \leq s_{j+1}$ (one large element per stripe).
\end{compactitem}
\item \klw: An instance with only \emph{one small element} in each small-element stripe.
\begin{compactitem}
\item Input: Unit-sized elements $s_{1},\ldots,s_{k + 1}$ \defn{sorted}, where $s_1=-\infty$ and $s_{k+1}=\infty$, and large elements $\ell_{1},\ldots,\ell_{\lw}$ (volume $\lw w$) \defn{unsorted}. 
\item Output: For each $\ell_{i}$, output its predecessor and successor in $S$.
\end{compactitem}
\item \kk: An instance with only \emph{one element in each stripe}, large or small. 
\begin{compactitem}
\item Input: Unit-sized elements $s_{1},\ldots,s_{k + 1}$ \defn{sorted}, where $s_1=-\infty$ and $s_{k+1}=\infty$, and large elements $\ell_{1},\ldots,\ell_{k}$ (volume $kw$) \defn{unsorted}. 
\item Output: The entire set in the sorted order. 
\end{compactitem}
\end{compactenum}

The format of lower bounds for \sk, \klw and \kk is as follows: let $X$ be the logarithm of the total number of permutations that an algorithm needs to achieve in order to solve the problem. As is easily observed, the values of $X$ for these three subproblems are $k \log (S/k)$, $\tilde{k} \log k$, and $k \log k$, respectively. Lemma~\ref{adversarylemma} below is the most technical part of this paper, and it helps us quantify the behavior of the adversary during small-block and large-element inputs for all three subproblems.  We use this lemma to prove the lower bounds for the individual three subproblems (found in \lemref{sklemma}, \lemref{klwlemma}, and \lemref{kklemma}). Then we put the three lemmas together to obtain the expression from \thmref{plelower}.

\begin{lemma}\label{adversarylemma}
Consider any algorithm for the \sk, \klw, or the \kk problem. There exists an adversary such that:
\begin{itemize}
    \item On the input of any block of $B$ short elements, the adversary answers comparisons between all elements in main memory such that the fan-out of this I/O is at most $2^B$. In other words, the number of permutations the algorithm needs to check is reduced by a factor at most $2^B$.
    \item On the input of any large element (costing $w/B$ I/Os), the adversary answers comparisons between all elements in main memory such that the fan-out of this I/O is at most $O(M)$. In other words, the number of permutations the algorithm needs to check is reduced by a factor at most $O(M)$.
\end{itemize}
\end{lemma}

\noindent
\textbf{Proof of Lemma~\ref{adversarylemma}:} We prove this lemma  by describing the adversary. We capture the information learned at every point of the algorithm by assigning a \defn{search interval} to every large element:

\begin{definition} [\defn{Search interval}]
A search interval $R(\ell)=(s_i,s_j)$ for a large element $\ell$
at step $t$ is the
narrowest interval of small elements where $\ell$ can possibly land in the 
final sorted order, given what the algorithm has learned so far.
\end{definition}

It will be useful to consider the binary tree $\mathcal{T}$ on the set $\setShort$. The search interval of any large element at any point during the execution of the algorithm is a contiguous collection of leaves in $\mathcal{T}$. Note that it can never be disconnected.

For simplicity we will assume that the size of $\setShort$ is a power of $2$, and hence $\mathcal{T}$ is perfectly balanced. Also, if $R(\ell)=(s_{i},s_{j})$ is the range of a large element, we will make sure the adversary ``rounds off'' the search space so that the new range corresponds exactly to a subtree of some node in $\mathcal{T}$. This is accomplished by first finding the least common ancestor $lca$ of $s_{i}$ and $s_{j}$,  and then shrinking the search space of $\ell$ to either the search space in the left subtree of $lca$ or to the search space in the right subtree of $lca$, whichever is larger. Thus each large element $\ell$ at any time has an associated node in $\mathcal{T}$, which we denote by $v(\ell)$. We also denote the interval corresponding to $v(\ell)$ (this is just the interval of its subtree) as $\intr(v(\ell))$. 

In the remainder, we find it convenient to work with logarithms of size of search spaces. For this purpose, we will use the term ``bit''. The learning of one ``bit'' by the algorithm corresponds to the halving of the search space of some large element.

\noindent\textbf{Mechanics of the adversary's strategy:} Our adversary will try to maintain the following invariant at all times during the execution of the algorithm.

\noindent\textbf{Invariant:} The search intervals of large elements in main memory are disjoint.

We denote by $\{\ell^{p - 1}_{i}\}_{i = 1}^{M/w}$ the set of at most $M/w$ large 
elements in memory before the $p$th I/O.
By hypothesis, the nodes in $\mathcal{T}$ belonging to the set
$\{v(\ell^{p - 1}_{i})\}_{i = 1}^{M / w}$ have no ancestor-descendant
relationships between them.
We write $S^{p - 1}_{i}$ to denote $\intr(v(\ell^{p - 1}_{i}))$, the 
search interval of large element $\ell^{p - 1}_{i}$ at step $p - 1$.

\noindent \textbf{Small-block input.} Consider the incoming block. 
We denote $n_{p,i}$ as the number of  incoming small elements that belong to $S^{p - 1}_{i}$.
These elements divide $S^{p - 1}_{i}$ into $n_{p,i} + 1$ parts
$\{P_{1}, \ldots, P_{n_{pi} + 1}\}$, some of them possibly empty.
The largest of these parts (say $P_{j}$) is of size at least $1 / (n_{p,i} + 1)$ 
times the size of $S^{p - 1}_{i}$.
The new search interval of $\ell^{p}_{i}$ is defined to be the highest node in $\mathcal{T}$ such that $\intr(v) \subset P_{j}$.

\noindent \textbf{Large element input.}
On an input of a large element $\ell^{p}_{\text{new}}$ (with search interval $S^{p - 1}_{\text{new}}$),
the adversary uses a strategy similar to that one on a small-block input to compare  
$\ell^{p}_{\text{new}}$ with the (at most)
$M$ small elements present in memory. These $M$ small elements divide 
$S^{p-1}_{\text{new}}$ into at most $M$ parts, and the new search interval
of $\ell^{p}_{\text{new}}$ corresponds to the highest node in $\mathcal{T}$ that contains the largest part.

This is the temporary search interval $S_{\text{new}}$, with the corresponding
node $v_{\text{new}}$.

$S_{\text{new}}$ can be related to the search intervals of large elements in memory in three ways:
\paragraph{Case 1.}
The element $\ell^{p}_{\text{new}}$ shares a node
with another large element $\ell^{p}_{i}$.
The conflict is resolved by sending $\ell^{p}_{\text{new}}$ 
and $\ell^{p}_{i}$ to the left and right children of $v_{\text{new}}$, 
respectively.
\paragraph{Case 2.}
The element $\ell^{p}_{\text{new}}$ has an ancestor in memory.
The ancestor is sent one level down, to the child that does not contain
$v_{\text{new}}$ in its subtree.
Thus the conflict is resolved while giving at most $O(1)$ bit.

\paragraph{Case 3.}
The element $\ell^{p}_{new}$ has descendants in memory.

Denote the nodes that are descendants of $v_{\text{new}}$ in $\mathcal{T}$ as
$v_{1}, \ldots, v_{M/w}$. Let the corresponding search intervals be $S^{p-1}_{1}, \ldots, S^{p-1}_{M/w}$, respectively. Let $X = \cup_{i = 1}^{M/w} S^{p-1}_{i}$ and $Y = S_{\text{new}} \setminus X$. The set $Y$ is a union of at most $M/w+1$ intervals, each of which we denote by $Y_{i}$.
Let $Z$ be the largest interval from the set $\{S^{p-1}_{1}, \ldots, S^{p-1}_{M/w}, Y_{1}, \ldots, Y_{M/w}\}$.
Hence, $ \vert Z \vert \geq \vert S_{\text{new}} \vert / (2M/w)$.

There are two cases to consider. The first case is when $Z = S^{p-1}_{i}$ for some $i$.
In this case, $S_{\text{new}} = S^{p - 1}_{i}$.
In doing this we have given at most $O(\log M)$ bits.
Now we proceed as in Case~1 to resolve the conflict with at most $O(1)$ extra bits.
Otherwise, if $Z = Y_{i}$ for some $i$, then the adversary allots $\ell^{p}_{\text{new}}$ to the highest node $v$ in $\mathcal{T}$ such that $\intr(v) \subseteq Z$.

\subsection{Analysis}

We have the following auxiliary lemmas:
\begin{lemma}
\lemlabel{klwlemma1}
On a small-block input, the adversary gives at most
$O(\log (n_{p,i} + 1))$ bits to $\ell^{p}_{i}$.
\end{lemma}

\begin{proof}
Observe that
\[
\| P_{j} \| \geq \frac{\|S^{p - 1}_{i}\|}{(n_{pi} + 1)}.	
\]

Divide $S^{p-1}_{i}$ into $2(n_{p,i} + 1)$ equal parts (with the last one being 
possibly smaller).
If $P_{j}$ is equal to the union of two consecutive such parts, there is a node 
in $\mathcal{T}$ corresponding to $P_{j}$, and the adversary has given exactly
$\log (n_{p,i} + 1)$ bits.
Otherwise, $P_{j}$ contains at least one of these parts, for which there is a 
node $\log (n_{p,i} + 1) + 1$ levels below $v(\ell^{p-1}_{i})$, which is 
how many bits the adversary gives in this scenario. 

In either case, the maximum number of bits given by the adversary is
$O(\log (n_{p,i} + 1))$, as claimed.\qed
\end{proof}

\begin{lemma}
\lemlabel{klwlemmashort}
On a small-block input, the adversary gives at most $O(B)$ bits.
\end{lemma}

\begin{proof}
This follows easily from \lemref{klwlemma1}.
Let $G$ denote the total number of bits given by the adversary during the input 
of a block of small elements.
It can be seen that $G = \sum_{i = 1}^{B} \left(\log(n_{p,i}+1)+1\right)$.
By definition $\sum_{i=1}^{B} n_{pi} = B$, implying that
$\sum_{i=1}^{B} \log(n_{p,i}+1) \leq B$, which in turn implies that
$G < 2B = O(B)$.\qed
\end{proof}

\begin{lemma}
\lemlabel{klwlemmalong}
During the input of a large element, the adversary gives at most $O(\log M)$ bits.
\end{lemma}

\begin{proof}
The number of bits given due to comparisons with small elements already in memory 
is $O(\log M)$.
In each of the three cases an additional $O(\log (M/w))$ bits 
are given.
Thus, the total number of bits given by the adversary during the I/O of a large 
element is $O(\log M)$.\qed
\end{proof}

\subsection{Putting It All Together: getting lower bounds for \sk, \klw and \kk}

\paragraph{1) \sk Lower Bound.}
The proof rests on the following action of the adversary: 
in the very beginning, the adversary gives the algorithm the extra information that the $i$th largest large element lies somewhere between $s_{(i-1)\alpha}$ and $s_{i\alpha}$, where $\alpha = \volShort/k$. In other words, the adversary tells the algorithm that the large elements are equally distributed across $\mathcal{S}$, one in each chunk of size $\volShort/k$ in $\mathcal{S}$.

This deems the invariant of large elements in main memory having disjoint search intervals automatically satisfied.

Because any algorithm that solves \sk must achieve $\Omega(k \log (S/k))$ bits of information, we have that 
\begin{lemma}
\lemlabel{sklemma}
$\sk = \Omega\left(\min \left(\frac{kw}{B}\log_{M} \frac{S}{k} , \frac{k}{B} \log \frac{S}{k} + \frac{kw}{B}\right) \right).$
\end{lemma}

\paragraph{2) \klw Lower Bound.}
To solve \klw, an algorithm needs to learn $k\log \tilde{k}$ bits of information. Using the adversary strategy we described, we obtain the following lower bound:

\begin{lemma}
\lemlabel{klwlemma}
\klw = $\Omega\left(\min \left(\frac{\lw w}{B} \log_{M} k, \frac{\lw}{B} \log k + \frac{\lw w}{B}\right)\right).$
\end{lemma}

\paragraph{3) \kk Lower Bound.}
To solve $k$-$k$, an algorithm needs to learn $k\log k$ bits of information. 
In the $k$-$k$ problem, we expect to produce the perfect interleaving of the small and large
elements in the final sorted order. That is, \emph{each element lands in its own leaf of $\TT$}.

Therefore, the adversary does not posses the freedom to route elements down the tree at all times
using the strategy we described. Instead, the strategy is used for a fraction of total bits the algorithm learns, and the remaining fraction is used to make up for the potential imbalance created by sending more elements to the left or to the right. We call these \emph{type one} and \emph{type two} bits, respectively. Late bits are effectively given away for free by the adversary.

More formally, we define the \defn{node capacity ($c^{T}(v)$)} as  the number  of large elements that pass through $v$ during the execution of an algorithm. If the $k$-$k$ algorithm runs in $T$ I/Os, then the node capacity of $v$ at a level $h$ of $\TT$ is designated by $c^{T}(v) = k / 2^{h}$.

\begin{definition}[\defn{type one} and \defn{type two bits}]
\deflabel{typesofbits}
A bit gained by a large element $\ell$ is an \defn{type one} bit if,
when $\ell$ moves from $v$ to one of $v$'s children, at most $c^{T}(v)/4 - 1$ other large
elements have already passed through $v$.  The remainder of the bits are
\defn{type two} bits.  
\end{definition}

Because a small-block input gives $O(B)$ bits and a large-element input gives $O(\log M)$ bits,
and we need to achieve all type one bits to solve the problem (there are $(k \log k) / 4$ of them), we obtain the following lower bound:


\begin{lemma}
\lemlabel{kklemma}
$\kk = \Omega\left(\min \left(\frac{kw}{B} \log_{M} k, \frac{k}{B} \log k + \frac{kw}{B}\right) \right).$
\end{lemma}

Now we combine everything to get the proof of our PLE and sorting lower bounds.

\noindent\textbf{Proof of \thmref{plelower}}

The lower bounds for \kk, \klw and \sk are each a minimum of two terms; it is safe to add the respective terms as the transition between which term dominates occurs at exactly the same value of $w$ for each of the subproblems. Adding the terms for the lower bounds of \kk and \sk provides the $\frac{k}{B} \log S$ and $\frac{kw}{B} \log_{M} S$ terms in \thmref{plelower}. Adding the terms for the lower bounds of \kk and \klw, and using that $k+\tilde{k} = L/w$ provides the $\frac{L}{wB}\log S$ and $\frac{L}{B}\log_{M}S$ terms in \thmref{plelower}.

\input{wlessthanb}

\vspace{-1mm}\section{Upper bounds on Sorting and the Batched Predecessor Problem}\label{BFS}

\vspace{-1mm}Our algorithm for $\sortN$ works in three steps:1) sort the short elements using traditional multi-way external memory merge-sort \cite{AggarwalVi88}, 2) solve the associated $\plesl$ problem, and 3) sort the long stripes obtained again using multi-way mergesort. The first and third steps give the first and third terms in the sorting complexity in Theorem 2.

We give two algorithms to solve $\plesl$: PLE-DFS and PLE-BFS. The final upper bound is the minimum of the two terms, as presented in \thmref{pleupper}.

\noindent\textbf{PLE-DFS:} PLE-DFS builds a static B-tree $\tree$ on $S$, and searches for large
elements in $\tree$ one by one.
This approach is preferred in the case of really large elements, and it is better
to input them fewer times.

We dynamically maintain a smaller B-tree $\tree^{\prime}$ that
contains only \defn{border elements} (the two small elements 
sandwiching each large element in the final sorted order)
and has depth at most $\log_{B}k$.
All large elements first travel down $\tree^{\prime}$ to locate their stripe.
Only those elements for which their stripe has not yet been discovered  need to
travel down $\tree$.
After a new stripe is discovered in $\tree$, it is then added to
$\tree^{\prime}$. The total cost becomes

\begin{equation}
\eqlabel{ple-dfs}
O\left(\frac{\volLong}{\lenLong}\log_{B} \numStripe + \numStripe \log_{B} S + \frac{L}{B} + \frac{S}{B}\right).
\end{equation}

\noindent\textbf{PLE-BFS:} Our second algorithm for PLE uses a batch-searching tree with fanout
$\Theta(M)$.
When a node of the tree is brought into memory, we route all large elements via
the node to the next level.
We process the nodes of the $M$-tree level by level so all large elements proceed at an
equal pace from the root to leaves.
This technique is helpful when large elements are sufficiently small so that
bringing them many times into memory does not hurt while they benefit from a
large fanout.

The analysis is as follows: at each level of M-tree, the algorithm spends
$\Theta(L/B)$ I/Os in large-element inputs.
Every node of the tree is brought in at most once, which results in total
$O(S/B)$ I/Os in small-element inputs.
The total number of memory transfers for PLE-BFS then becomes
\begin{equation}
\eqlabel{ple-bfs}
O\left(\frac{L}{B} \log_{M} \volShort + \frac{\volShort}{B}\right).
\end{equation}

Our final upper bound is the better of the two algorithms:

\begin{theorem}[\textbf{PLE Upper Bound}]
\thmlabel{pleupper}
$$\plesl = O\left( \min\left\{ \bfsub, \dfs \right\}\right).$$
\end{theorem}

Substituting the lower and upper bounds of the batched predecessor problem ($\plesl$) derived in Theorems 3 and 4 into the complexity of sorting in Theorem 2 gives us lower and upper bounds on the sorting problem $\sortN$.

\noindent\textbf{Remark 1:} One observes that in Theorem 4 ($\plesl$ lower bound), the transition between the two terms in the minimum happens at $w = B \log M$. This is because when large elements are very large, the bound obtained by algorithms that do not input the large elements too often (PLE-DFS) is smaller than algorithms that input large elements multiple times (e.g., PLE-BFS).

\noindent\textbf{Remark 2:} The upper and lower bounds on $\plesl$ are tight for a wide range of parameters. Moreover, if the first and third terms in the complexity of sorting (Theorem 2) dominate the complexity of the associated $\plesl$ problem, our sorting algorithms are tight. 

\noindent\textbf{Remark 3:} We would like to draw the reader's attention to the second terms in the lower and upper bounds of $\plesl$: 

\[ \frac{k}{B}\log S+ \frac{L}{wB}\log k \ \ \ \text{versus} \ \ \ k\log_{B} S+ \frac{L}{w}\log_{B}k \]

The gap appears because while our algorithm (the two-tree PLE-DFS) works on $B$-trees and gets a fanout of $B$ per I/O, our lower bound only forbids fanouts larger than $2^B$. One may wonder whether an upper bound of $B$ on the fan-out is possible, as is the case with almost all searching problems in external memory. Consider the perfectly interleaved case, i.e., $n$ short and long elements each, and $k=n$ (the ``nuts and bolts'' version)). We show that if $n$ is very large, \textit{there is an algorithm that achieves a fanout of $2^B$}! This algorithm does not exist for small $n$, but it nevertheless shows that obtaining an unconditional upper bound of a fanout of $B$ is not possible. 

\begin{theorem}\label{slices}
Consider the problem where $k$ short elements are given sorted, $k$ large elements each of size $w \geq B$ are given unsorted, and it is given that in the final sorted order the elements are perfectly interleaved. There exists $k_{0} \in \mathbb{N}$ such that for all input sizes $k > k_{0}$, there is an algorithm that after $O\left(\frac{k}{B} \log_{M/B} \frac{k}{B} \right)$ I/Os in preprocessing outputs the sorted order in $O \left( \frac{k}{B} \log k + \frac{kw}{B}) \right)$ I/Os.
\end{theorem}

\begin{proof}
We will assume that $w \geq \log M$. If not, then $ (k \log k)/B > (kw/B) \log_{M} k$, and we already have an upper bound (the BFS algorithm in Section $6$) that has complexity $O((kw/B) \log_{M} k)$ I/Os. Also for simplicity, we will assume that $k$ and $B$ are powers of two.

Since we are looking for an upper bound of $O\left( (k \log k)/B + kw/B \right)$, we need to achieve roughly a fan-out of $2^B$ per I/O. The basic idea is the following: assume there are $B$ large elements in memory, and their search spaces are $S_{1},\cdots S_{B}$. If an incoming short block has the medians of all the $S_{i}$s, then the input of this short block reduces every $S_{i}$ by a factor of $2$, and we get the desired $2^B$ fan-out. Of course, for this to continue, we would need the appropriate short block (containing the medians of the new search spaces, and so on). Thus, it is intuitively clear how to achieve the upper bound if one were allowed, say ${k \choose B}$ preprocessing. This is huge, and the main question is whether we can reduce it to $O\left( (k \log k)/B + kw/B \right)$.

To describe our algorithm, we will need a smaller data structure first, which we explain next.

\noindent\textbf{$2^{B}$ Tree:}
A $2^{B}$ tree for $\alpha$ levels on a sorted set $\mathcal{A}$ of $A$ unit-sized elements (denoted as $\tree(\mathcal{A},2^{B},\alpha)$) is a tree that performs the following : Assume $B$ long elements have to find their positions among elements in $\mathcal{A}$, and that they are in memory. Initially they could be anywhere (so their search space size is $A$). $\tree(\mathcal{A},2^{B},\alpha)$ is a data structure that reduces their search spaces to size $A/(2^{\alpha}B)$ using $\alpha + 1$ short block I/Os.
We briefly describe how to build this tree.
In the first step, we bring in the root block of the $B$-tree on $\mathcal{A}$, achieving a fan-out of $B$ for every long element.
 Partition the set $\mathcal{A} : = \{e_{1},...e_{A}\}$ into $B$ equally sized (sorted) subsets $\mathcal{A}_{i} = \{e_{(iA/B)+1},...e_{(i+1)A/B}\}$ ($0 \leq i \leq B-1$). Put the $B$ middle elements ,$\{e_{(2j+1)A/2B}\}_{j=0}^{B-1}$, into a block, which serves as the root of $\tree(\mathcal{A},2^{B},\alpha)$. Assume that the $i$th long element points to $A_{i}$ (so $A_{i}$ is its current search space). Upon comparison of this root block with the $B$ long elements in memory, each long element's search space is reduced by a factor of $2$, the total fan-out being $2^{B}$. In each such permutation a long element's search space is now either the left half or the right half of its original. For each of the $2^{B}$ permutations, make a block of $B$ short elements comprising of the middle pivots of the new search spaces corresponding to the permutation. These are the immediate children of the root node. We recurse on these nodes now and stop when we have built $\alpha$ levels of this tree.

\paragraph{Preprocessing Phase}

\begin{enumerate}

\item Define 
\[j = \frac{\log \log_{M/B} k/B - \log B}{B-1}\] and $g = (\log k - j) / j$.

Build $\tree(k, 2^{B} , j)$ on the sorted set of the $k$ short elements.\textit{For the algorithm to work, $j = \omega(1)$, which automatically puts a restriction on $k$. This is the $k_{0}$ referred to in the statement of the observation, and the algorithm works only if $k>k_{0}$.} Also, we will assume that $k$ is sufficiently large that $Bw < \log \log_{M/B} k$.
 
 \item Let $\{\mathcal{U}_{i}\}_{i=1}^{g}$ be a collection of sets, where $\mathcal{U}_{i}$ is the (sorted) set of all nodes at depth $i.j$ (so $i.j$ levels from the root) in the binary tree on $k$ short elements. Let $U_{i}$ be the cardinality of $\mathcal{U}_{i}$ ( $U_{i} = 2^{ij}$).  For every $\mathcal{U}_{i}$ do the following :
 
 \begin{itemize}
 \item Starting from the left, divide the set $\mathcal{U}_{i}$ into groups of size $B$.
  
 \item Let $G_{ir} = \{v_{1},...,v_{B}\}$ be such a group of nodes ($1 \leq r \leq U_{i}/B$). Let $\mathcal{A}_{ir}$ be the union of the search spaces of these nodes (union of the leaves of their subtrees). 
 
 \item Build the trees $\tree(\mathcal{A}_{ir},2^{B},j)$ for all groups $G_{ir}$ above.
 \end{itemize}
 
\end{enumerate}

\paragraph{Querying Phase}

Here we describe how to perform the query search :
\begin{enumerate}

\item Divide the set of $k$ large elements into groups of $B$ long elements arbitrarily. For every group, bring it in memory and perform the search on $\tree(k, 2^{B} , j)$ (sending every large element $j$ levels down the binary tree on $k$).
\item If all large elements have been flushed at least $ij$ levels down (where $1\leq i \leq g$), for each group $G_{ir}$ in $\mathcal{U}_{i}$, do the following :

\begin{itemize}

\item Find the set of large elements pointing to a node in $G_{ir}$ (in other words, all large elements which have been found to belong to $\mathcal{A}_{ir}$, which, by definition, is the union of the search spaces of nodes in $G_{ir}$). Let this set be $\mathcal{Q}$, with $Q$ long elements. \textit{Note that by the definition of the \kk problem (large and small elements perfectly interleaved), $Q = A$, i.e., there are as many  large elements pointing to a node in $G_{ir}$ as the number of small elements/leaves in the subtrees rooted at nodes in $G_{ir}$.}

\item Divide $\mathcal{Q}$ into groups of size $B$, bring each group into memory one at a time and flush it through $\tree(\mathcal{A}_{ir},2^{B},j)$.
\end{itemize}

\end{enumerate}

\begin{lemma}[Complexity Analysis]
The I/O complexity of the above algorithm is $O((k\log k)/B + kw/B)$ I/Os.
\end{lemma}

\begin{proof}
Each short block I/O gets a fan-out of $2^B$, as it halves the search space of $B$ large elements in memory. The total fan-out required is $k!$, so the number of short block I/Os is $O((k \log k)/B$.

The large elements are swiped $g$ times (once after every $j$ levels). Each time costs $kw/B$ I/Os, requiring a total of $O\left(kw\frac{\log k}{\log \log_{M/B} k}\right)$ I/Os (by definition of $g$), which is $O(kw/B)$ since $k$ is sufficiently big to guarantee $Bw < \log \log_{M/B} k$.
\end{proof}

\paragraph{Preprocessing Analysis}

\begin{lemma}
Starting with a binary tree on $\mathcal{A}$, $\tree(\mathcal{A},2^{B},\alpha)$ can be built in $O(B.2^{\alpha B})$ I/Os.
\end{lemma}

\begin{proof}
There are $O(2^{\alpha B})$ blocks in $\tree(\mathcal{A},2^{B},\alpha)$. Each block takes a maximum of $B$ I/Os, since all the levels of the binary tree are already built and one only needs to go one level down the binary tree in order to bring the middle pivot.

\end{proof}

The number of trees of type $\tree(\mathcal{A},2^{B},j)$ built by our algorithm can be bounded by the last level. There are $(k/B)/2^j$ nodes in the last level on which this data structure is built. Thus the total cost is bounded by 

\[ (B.2^{jB})\frac{k}{2^{j}B} = k2^{j(B-1)} = \frac{k}{B} \log_{M/B} \frac{k}{B}, \]

where the last inequality follows by the definition of $j = \frac{\log \log_{M/B} k/B - \log B}{B-1}$. This finishes the proof of the preprocessing claim.\qed
\end{proof}

%% file: wlessthanb.tex
\subsection{Generalization of lower bounds to the case when $w_{1}<w_{2}<B$}\label{allcases}
So far, our assumptions on the record sizes accommodate one set of records of unit size, and the other set contains items larger than a block. But what if we have two record sizes, where both can be relatively large but still smaller than a block? In this section, we generalize the lower bound results to this case.

The number of bits required by an algorithm remains unchanged as that is an information-theoretic lower bound. It remains to see how the invariant maintained by the adversary limits the information achieved by any algorithm.

The input of a small block contains $B/w_{1}$ elements now. Since the large elements in memory have disjoint search spaces, the maximum number of bits achievable by this I/O is $B/w_{1}$, which is the case when each of these small elements is a pivot for a unique large element. Thus we get $O(B/w_{1})$ bits per I/O.

The input of a large block contains $B/w_{2}$ large elements. The memory can contain at most $(M-B)/w_{1}$ small elements, and so the total number of possible permutations achievable is

\begin{eqnarray}
 P &=& {\frac{M-B}{w_{1}}+\frac{B}{w_{2}} \choose \frac{B}{w_{2}}} \left(\frac{B}{w_{2}}!\right) \notag \\
 &<& {\frac{M}{w_{1}} \choose \frac{B}{w_{2}}} \left(\frac{B}{w_{2}}!\right) \notag \\
 &<& \left( \frac{ew_{2}M}{Bw_{1}} \right)^{B/w_{2}} \left(\frac{B}{w_{2}}!\right) \notag 
\end{eqnarray}
 This gives 
 \begin{eqnarray}
 \log P &=& O \left( \frac{B}{w_{2}} \log \left( \frac{Mw_{2}}{Bw_{1}} \right) + \frac{B}{w_{2}}  \log \left(\frac{B}{w_{2}} \right) \right) \notag \\
 &=& O \left( \frac{B}{w_{2}} \log \left( \frac{M}{w_{1}} \right)\right) \notag
 \end{eqnarray}
 bits per I/O.
 
In both cases, the amortized number of bits achieved is:
\begin{enumerate}
\item $O(B/w_{1})$ bits per I/O, equivalent to $O(1)$ bit per $w_{1}/B$ I/Os.
\item $ \frac{B}{w_{2}} \log \left( \frac{M}{w_{1}} \right)$ bits per I/O, equivalent to $\log(M/w_{1})$ bits per $w_{2}/B$ I/Os.
\end{enumerate}
\qed

%% file: conclusion.tex
\vspace{-1mm}\section{Conclusion and Open Problems}

\vspace{-1mm}We derived upper and lower bounds on sorting and batched predecessor in the RAM and DAM models, when comparison or I/O costs depend on the length of the items being compared. In many settings, we show that the optimal sorting algorithm involves the optimal batched predecessor problem as a subroutine, and develop algorithms for the batched predecessor problem. 

While our results are for the two-size setting, we would like to point out that our algorithms generalize to the multiple-sizes setting. However, generalizing our lower bound techniques to the multiple-size setting requires more ideas.